\documentclass[a4paper]{article}

\usepackage[T1]{fontenc}
\usepackage{graphicx}
\usepackage{adjustbox}
\usepackage{diagbox}
\usepackage{multirow}
\usepackage{subcaption,booktabs}
\usepackage{fancyvrb}
\fvset{frame=single,framesep=1mm,fontfamily=cmtt,framerule=.3mm,numbersep=1mm,numberblanklines=false,commandchars=\\\{\},numbers=left}
\usepackage{tabularx}
\usepackage{amsthm}
\usepackage{mathtools}
\usepackage[boxed,noend]{algorithm2e}
\let\chapter\undefined

 \newtheorem{theorem}{Theorem}[section]
\newtheorem{proposition}[theorem]{Proposition}
\newtheorem{lemma}[theorem]{Lemma}

\newtheorem{observation}[theorem]{Observation}
\newtheorem{example}[theorem]{Example}
\newtheorem{definition}[theorem]{Definition}

\newcommand{\name}[1]{\mathrm{#1}}
\newcommand{\funct}[2]{\mathrm{#1}\!\left({#2}\right)}
\newcommand{\To}{\!\rightarrow\!}
\newcommand{\card}[1]{\left|{#1}\right|}

\renewcommand{\epsilon}{\varepsilon}
\renewcommand{\phi}{\varphi}
\newcommand{\oof}[1]{\mathrm{O}\!\left({#1}\right)}
\newcommand{\oofi}[1]{\oof{#1}}

\newcommand{\textcd}[1]{\textup{\texttt{\small{#1}}}}
\newcommand{\true}{\textrm{True}}
\newcommand{\false}{\textrm{False}}
\newcommand{\vars}[1]{\funct{vars}{#1}}

\usepackage{tikz}
\usepackage{xcolor}

\begin{document}
\title{Worst-case Analysis for\\Interactive Evaluation of Boolean Provenance}

\clubpenalty=10000
\widowpenalty = 10000

\author{Antoine Amarilli\\{\normalsize LTCI, Télécom Paris, Institut Polytechnique de Paris}
	\and 
	Yael Amsterdamer\\{\normalsize Bar-Ilan University}
}
\date{}

\maketitle

\begin{abstract}

In recent work, we have introduced a framework for fine-grained consent management in databases, which combines Boolean data provenance with the field of interactive Boolean evaluation. In turn, interactive Boolean evaluation aims at unveiling the underlying truth value of a Boolean expression by frugally probing the truth values of individual values. The required number of probes depends on the Boolean provenance structure and on the (a-priori unknown) probe answers. Prior work has analyzed and aimed to optimize the \emph{expected} number of probes, where expectancy is with respect to a probability distribution over probe answers. This paper gives a novel \emph{worst-case} analysis for the problem, 
inspired by the decision tree depth of Boolean functions.

Specifically, we introduce a notion of \emph{evasive provenance expressions}, namely expressions, where one may need to probe all variables in the worst case. We show that read-once expressions are evasive, and identify an additional class of expressions (acyclic monotone 2-DNF) for which evasiveness may be decided in PTIME. As for the more general question of finding the optimal strategy, we show that it is coNP-hard in general. We are still able to identify a sub-class of provenance expressions that is ``far from evasive'', namely, where an optimal worst-case strategy probes only $\oof{\log{n}}$ out of the $n$ variables in the expression, and show that we can find this optimal strategy in polynomial time.
\end{abstract}

\section{Introduction}\label{sec:intro}
There is a large body of work on Boolean provenance for database queries (e.g.,~\cite{amsterdamer2011limitations,amsterdamer2011aggregate,buneman2001why,cheney2009provenance,deutch2014circuits,green2007provenance,imielinski1984incomplete,olteanu2012factorised}). The basic paradigm is that one starts by associating a Boolean variable with every tuple in the input database; then, for queries in a given language (say, Select-Project-Join-Union), we explain how the query operators can be extended to annotate their results with Boolean expressions depending on the input tuples. For instance, for the join operation, a tuple obtained by the joining two tuples annotated by~$x$ and~$y$ will be annotated by~$x \wedge y$. If we union two relations and the same tuple occurs in both with annotations by~$x$ and~$y$, then the tuple in the (set) union of the two relations will be annotated by~$x \vee y$. A basic property of such provenance structures, dating back to~\cite{imielinski1984incomplete}, is that these provenance expressions describe the \emph{Possible Worlds} that make the query true: a truth valuation~$v$ of the Boolean variables satisfies the provenance expression annotating a tuple~$t$ if and only if~$t$ appears in the query result when running the query over the sub-database consisting only of the tuples annotated by variables that~$v$ maps to \true{}.

Building on these provenance structures, we have proposed in recent work a framework for consent management in shared databases~\cite{drien2021managing}. The setting is such that multiple peers share data but do not necessarily agree upfront to every kind of use of their data -- rather, explicit permission is needed when the data is to be used in a particular way, e.g., shared with a third party. This gets more complex when data is transformed via queries: when we wish to use a query result, it is challenging to determine whose consent should be asked, because the result may rely on multiple tuples originally contributed by multiple peers. The model proposed in~\cite{drien2021managing} follows the same possible worlds semantics that we have mentioned above: namely, consent is granted with respect to an output tuple~$t$ if and only if, on the sub-database of the input database containing all tuples for which consent was granted, the output~$t$ is part of the query result. This means that the problem of consent management reduces to the problem of \emph{Interactive Evaluation of Boolean provenance expressions}, i.e., that of determining the truth value of these expressions by probing individual variables to reveal their truth values. 
Interactive Boolean Evaluation has been extensively studied, and has mostly focused on the following problem: given probabilities on the truth values of the Boolean variables (namely, the probability of obtaining an affirmative probe answer for each variable), we must determine a probing strategy (i.e., which variables to probe and in what order) so that the expected number of probes is minimized. There is a significant body of work on the topic (e.g.,~\cite{allen2017evaluation,amsterdamer2019pepper,boros2000sequential,cicalese2014diagnosis,deshpande2014approximation,drien2021managing,kaplan2005learning,keren2008reduction,unluyurt2004sequential}) and in~\cite{drien2021managing} we have shown how to leverage this work for our application of consent management.

In this paper, we start from Boolean provenance and study the complexity of interactively evaluating it, but focus on a different optimization goal. Namely, rather than optimizing the expected number of probes, we focus on optimizing the worst-case number of probes, i.e., minimizing the maximal number of probes that may be needed on \emph{any} possible sequence of probe answers. In contrast to expected-case analysis, worst-case analysis requires no probability estimation for probe answers, and is useful when one aims for ``cautious'' strategies in terms of the number of probes.

\begin{example}\label{ex:intro}
Consider a database with two tuples, one has the Boolean provenance $x\land y$ and the other has the provenance $x\lor z$ (for now, regardless of how this database could be created). If we probe~$x$ first, i.e., discover its truth value, we will make at most two probes in total: if~$x$ is \true{}  we need not probe~$z$ to evaluate the second expression; and if~$x$ is \false{} we need
not probe~$y$ to evaluate the first expression. This strategy optimizes the worst-case number of probes for the given provenance. In contrast, a non-optimal strategy would probe~$y$ and~$z$ first and in the worst case would require additionally querying~$x$ to reach full evaluation.
\end{example}

A natural way to capture interactive evaluation strategies is by representing the Boolean expression as a 
\emph{Binary Decision Diagram}~\cite{bdd}. Intuitively, BDDs are acyclic graphs where each (inner) node corresponds to a variable to be probed next, where edges correspond to possible probe answers, 
and where leaves give the result of evaluating the function. The worst-case number of probes for a given strategy is the depth (i.e.\ longest root-to-leaf path) of the BDD representation of the strategy. Worst-case optimization for interactive evaluation then amounts to finding a BDD with minimal depth for the given provenance expressions. We can always find a trivial BDD where this depth is the number of variables of the expression, by asking about all variables, but in some cases we may be able to achieve a lower depth.

We define two related decision problems. We first ask: can we tell, by looking at a given provenance expression, whether the worst-case number of probes is simply the number of variables in the expression; or equivalently, is there no BDD for the expression whose depth is smaller than the number of variables? 
Following existing work on Boolean functions~\cite{yao1988monotone,chronaki1990survey},
we call \emph{evasive} the provenance expressions that have this trait, i.e., that require probing all their variables in the worst case, and call the decision problem \texttt{DEC-BDD-EVASIVE}. Evasiveness implies a lower bound for the worst-case performance of probing strategies: if a provenance expression is evasive, then there is no hope of improving on the trivial BDD that queries all variables in some order. 
Second, we introduce the problem \texttt{OPT-BDD-DEPTH} of finding the optimal strategy, namely the BDD with minimal depth for given provenance expressions. We also study its phrasing as a decision problem, called \texttt{DEC-BDD-DEPTH}, where we must  decide whether the depth of the optimal BDD for the expression (which we call \emph{the expression depth}) is less than some given $k$: this problem generalizes \texttt{DEC-BDD-EVASIVE}. 
We study these problems for multiple classes of provenance expressions of interest. Our main complexity results are the following.

\paragraph*{General Provenance Expressions} We start by considering general provenance expression, which in particular may involve negation (this can be generated for queries with negation/difference). We show that \texttt{DEC-BDD-DEPTH} and \texttt{OPT-BDD-DEPTH} are coNP-hard. Since these problems are intractable in the most general settings, we then turn to analyzing sub-classes of interest.  

\paragraph*{Read-once Provenance} Multiple lines of previous work have established the importance of \emph{read-once provenance}, namely expressions where no variable occur more than once. Generalizing to multiple provenance expressions (of multiple output tuples), we may distinguish between sets of expressions that are \emph{overall read-once} (namely no variable occur twice even across expressions)  and sets where each individual expression is read-once but variables may re-occur in different expressions. For overall read-once provenance expressions, or ones that may be written in an equivalent overall read-once form, we show that they are always evasive.  This means that we cannot improve the worst-case probing complexity of these expressions.
By contrast, we show that read-once provenance is not always evasive when it is not overall read-once. 

\paragraph*{Monotone Provenance} A class of queries whose provenance is often studied is that of SPJU (Select-Project-Join-Union queries). For such queries, the provenance includes no negation and may be represented as a monotone DNF formula. It is open whether our decision problems are tractable for this class, and we show that the problem is far from trivial. Specifically, we show that there exists a sequence of such provenance expressions that are ``far from evasive'', namely for which there is an optimal probing strategy that is exponentially better than the trivial strategy:
it only probes~$\oof{\log{n}}$ variables in the worst case, where~$n$ is the number of variables in the expression. For this particular sequence, the optimal strategy may be efficiently computed. 

\paragraph*{Acyclic monotone 2-DNF Provenance}
 We finally identify a sub-class of interest, which we call \emph{acyclic monotone 2-DNF}. These are monotone 2-DNF formulas, where the graph obtained by treating each term as an edge between its variables is acyclic. For this class, we show a polynomial-time algorithm to check whether a given expression is evasive. The algorithm is based on a non-evasiveness pattern, namely a pattern that occurs in the graph if and only if the corresponding expression is non-evasive.

The rest of this paper is organized as follows. We start with an overview of our model in Section~\ref{sec:model}. Our formal results and proofs are detailed in Section~\ref{sec:results}. We overview related work in Section~\ref{sec:related} and conclude in Section~\ref{sec:conc}.

\begin{table*}
\setlength{\tabcolsep}{3pt}
{\small
\begin{tabularx}{0.46\textwidth}{lll|l}
\multicolumn{4}{l}{\textbf{Acquisitions}}\\                                \toprule
Acquired & Acquiring & \multicolumn{2}{l}{Date} \\
\midrule
A2Bdone & Zazzer & 7/11/2020 & $a_0$\\
microBarg & Fiffer & 1/5/2017 & $a_1$ \\
fPharm & Fiffer & 1/2/2016 & $a_2$ \\
Optobest & microBarg & 8/8/2015 & $a_3$ \\
\bottomrule
\multicolumn{4}{l}{~}  \\
\multicolumn{4}{l}{~}  \\
\end{tabularx}\hspace{2mm}\begin{tabularx}{0.48\textwidth}{lll|l}
\multicolumn{4}{l}{\textbf{Education}}\\
\toprule
Alumni & Institute & \multicolumn{2}{l}{Year}\\
\midrule
Usha Koirala & U.\ Melbourne & 2017 & $e_0$  \\
Pavel Lebedev & U.\ Melbourne & 2017 & $e_1$  \\
Nana Alvi & U.\ Sau Paolo & 2010 & $e_2$ \\
Nana Alvi & U.\ Melbourne & 2017 & $e_3$  \\
Gao Yawen & U.\ Sau Paolo & 2010 & $e_4$ \\
Amaal Kader & U.\ Cape Town & 2005 & $e_5$  \\
\bottomrule
\end{tabularx}\hspace{2mm}\begin{tabularx}{0.62\textwidth}{lll|l}
\multicolumn{4}{l}{\textbf{Roles}}\\
\toprule
Organization & Role & \multicolumn{2}{l}{Member} \\
\midrule
A2Bdone & Founder & Usha Koirala & $r_0$ \\
A2Bdone & Founding member & Pavel Lebedev & $r_1$ \\
A2Bdone & Founding member & Nana Alvi & $r_2$ \\
microBarg & Co-founder & Nana Alvi & $r_3$\\
microBarg & Co-founder & Gao Yawen & $r_4$ \\
microBarg & CTO & Amaal Kader & $r_5$ \\
\bottomrule
\end{tabularx}
}\caption{Example annotated database.}
        \label{tab:db}
        \vspace{-7mm}
\end{table*}

\section{Model}\label{sec:model}

We introduce the model for the problems that we study in this paper. We 
start by recalling the notion of Boolean provenance as well as a representation form for 
Boolean expressions called Binary Decision Diagrams that will be useful in our proofs. We then define the decision problems that we study in the remainder of the paper.    

\subsection{Boolean Provenance}

We assume familiarity with standard relational database terminology~\cite{abiteboul1995foundations}. 
Let $X$ be a set of Boolean variables. We recall the notion of an annotated database~\cite{green2007provenance,imielinski1984incomplete}, a relational database where each tuple is annotated by a variable from~$X$ (also called $X$-database). In our context, the Boolean variables stand for tuple consent.

\begin{definition} [Annotated database]
	\label{def:xdatabase}
An \emph{annotated database} $\bar{D} = (D, L)$ consists of a relational database $D$ and of a labeling function $L:\funct{tuples}{D} \mapsto X$  mapping each tuple $t$ in $D$ to a variable $\funct{L}{t}\in X$. We will denote $n \coloneqq \card{X}$.
\end{definition}

The premise is that there is a hidden truth associated with each variable, i.e., a ground truth correctness for the corresponding tuple.

\begin{definition}[Truth Valuation]\label{def:val}
A \emph{valuation} is a function $\name{val}:X \mapsto \{\true{},\false{}\}$.  Given an annotated database $\bar{D} = (D, L)$, denote by $D_{\name{val}}\coloneqq \{t\in D \mid \funct{val}{L(t)}=\true{}\}$.
\end{definition}

\begin{example}\label{ex:tab}
Table~\ref{tab:db} outlines an annotated database with three relations: \emph{Acquisitions}, including data on companies acquired by other companies; \emph{Roles}, including data on  roles of different organization members; and \emph{Education}, including data on university alumni. The right-most column shows the variable in $X$ annotating the tuple, and standing for the event that consent is granted for using this tuple (in response to a specific usage request). 
\end{example}

The provenance annotations of the input database can be propagated to the query output, so that Boolean expressions annotating tuples in the query output reflect the possible worlds of the input annotated database for which the output tuple appears in the query result. In particular, Select-Project-Join-Union (SPJU) queries are known to yield provenance in the form of monotone Boolean expressions (without negation), which are computable in DNF (i.e., as disjunctions of conjunctions) in PTIME~\cite{imielinski1984incomplete}.  When relational difference is allowed (SPJUD queries), the provenance expressions may also include negation~\cite{amsterdamer2011limitations} and may no longer be in DNF.

\begin{figure}
{\footnotesize
\begin{Verbatim}
SELECT DISTINCT a.Acquired, e.Institute
FROM Acquisitions AS a, Roles AS r, Education AS e
WHERE a.Acquired = r.Organization AND 
      r.Member = e.Alumni AND a.Date >= 2017.01.01 AND
      r.Role LIKE '\%found\%' AND e.YEAR <= year(a.Date)
\end{Verbatim}
}
\vspace{-3mm}
\caption{Query over the example database}
\label{fig:query}
\vspace{-3mm}

\end{figure} 

For a Boolean expression $\phi$ using only variables from $X$, we will denote by $\funct{val}{\phi}$ the truth value resulting from replacing each occurrence of a variable $x$ in $\phi$ with its truth value $\funct{val}{x}$. We can then define:

\begin{definition} [Provenance for Query results]
	\label{def:queryres}
	Given an annotated database $\bar{D} = (D, L)$ and a query $Q$ over $D$, the query result is $Q(\bar{D})=(Q(D), L')$ where $Q(D)$ is the standard query result and $L':Q(D)\To\name{Bool}[X]$ where $\name{Bool}[X]$ is the semiring of all Boolean expressions over~$X$~\cite{amsterdamer2011limitations,green2007provenance}. For every $t\in Q(D)$, and every valuation $\name{val}$ over $X$, we require that the expression $\phi=L'(t)$ is such that $\funct{val}{\phi}=\true{}$ if and only if $t\in D_{\name{val}}$ (Def~\ref{def:val}).
\end{definition}

\begin{example}\label{ex:query}
Consider the SPJU query in Figure~\ref{fig:query} over the DB in Table~\ref{tab:db}, which returns companies acquired since~2017 along with institutes in which founders of these companies had studied. The query results are shown in Figure~\ref{tab:res} along with their Boolean provenance expressions. In the event that the first two tuples in the Acquisitions relation are assigned $\false{}$ ($\funct{val}{a_0}=\funct{val}{a_1}=\false{}$) there are no query results (With consent) since the join of the Acquisitions table with the other tables would necessarily be empty. If, alternatively, the first tuple of each of the three relations is assigned $\true{}$ (i.e., $\funct{val}{a_0}=\funct{val}{r_0}=\funct{val}{e_0}=\true{}$), the query first result  tuple \textcd{(A2Bdone, U.\ Melbourne)}, derived from these input tuples, has consent.
\end{example}

\subsection{Interactive Evaluation via BDDs}
We focus here on \emph{evaluating} provenance expressions, namely revealing the truth values of expressions by probing individual variables to gradually reveal the underlying valuation. In the settings of \cite{drien2021managing}, probes correspond to asking tuple owners for their consent to use a tuple that they have contributed to a shared database. Naturally, the choice of which variables to probe is non-trivial, and different sequences of probes may lead to revealing the truth values of the provenance expressions much faster than others. 
\begin{example}
 Consider the first tuple of the query result in Table~\ref{tab:res}. Probing $a_0$ corresponds to asking consent to use the first tuple of the Acquisitions relation. If the answer is negative, then $\funct{val}{a_0}=\false{}$
 and we can conclude that the first result tuple is assigned $\false{}$ (in the context of consent management, we do not have consent to use this tuple). If by contrast $\funct{val}{a_0}=\true{}$ then we still do not know the consent status with respect to the result tuple, and we need to continue probing.  
There are many other probing strategies, e.g., we may start by probing $\funct{val}{a_1}$.
\end{example}

\begin{table}
\setlength{\tabcolsep}{3pt}
{\small
\begin{tabularx}{\columnwidth}{ll|l}
                                \toprule
                                Acquired & \multicolumn{2}{l}{Institute} \\
                                \midrule
                                A2Bdone & U.\ Melbourne & $(a_0\!\wedge\! r_0\!\wedge\! e_0)\vee (a_0\!\wedge\! r_1\!\wedge\! e_1) \vee (a_0\!\wedge\! r_2\!\wedge\! e_3)$\\
                                A2Bdone & U.\ Sau Paolo &  $ (a_0\!\wedge\! r_2\!\wedge\! e_2)$ \\
                                microBarg & U.\ Melbourne & $(a_1\!\wedge\! r_3\!\wedge\! e_3)$ \\
                                microBarg & U.\ Sau Paolo & $(a_1\!\wedge\! r_3\!\wedge\! e_2)\vee (a_1\!\wedge\! r_4\!\wedge\! e_4)$ \\
                                \bottomrule
                        \end{tabularx}
        \caption{Result of the example query.}
        \label{tab:res}
        }
        \vspace{-7mm}

\end{table}

A deterministic strategy of which variable to probe, based on the truth values revealed so far, can be generally modelled in terms of a Boolean Decision Diagram (BDD). In a BDD, each inner node stands for the current variable to probe and its outgoing edges reflect the strategy for the next steps if the observed value (probe answer) is \true{} and \false{} respectively. The leaves reflect the truth value by the end of the evaluation.  We state the definition in full generality for a set of Boolean functions.

\begin{definition}[Binary Decision Diagram (BDD)]\label{def:bdd}
A Boolean expression $\phi\in\name{Bool}[X]$ is \emph{constant} equal to $b$, for $b \in \{\true{}, \false{}\}$, if for every valuation $\name{val}$ of $X$, we have $\funct{val}{\phi} = b$.

For $x \in X$ and $b \in \{\true{}, \false{}\}$, the \emph{instantiation of~$x$ to~$b$  in~$\phi$}, written $\phi_{x=b}$, is the Boolean function on $X \setminus
  \{x\}$ obtained by assigning~$b$ to (all occurrences of) $x$ in $\phi$. We generalize this to a set of Boolean functions $\Phi$ by $\Phi_{x=b}=\{\phi_{x=b}\mid \phi\in\Phi\}$.

A BDD on a non-empty set of Boolean functions $\Phi\subseteq \name{Bool}[X]$ is a labeled directed acyclic graph (DAG) $G=(V,E,L_V)$ and distinguished \emph{root node} $v \in V$, where $V$ is a set of nodes labeled by $L_V:V\To X\cup (\Phi\To\{\true{},\false\})$ and $E\subseteq V\times V\times \{\true{},\false{}\}$ is a set of directed edges labeled by a Boolean value. The DAG is inductively defined as follows.
\begin{itemize}
    \item If $\Phi$ has at least one variable~$x$, a BDD for $\Phi$ may consist of a root node $v$ with $L_V(v) = x$, an outgoing edge labeled by \true{} and connecting $v$ to the root of a BDD for $\Phi_{x=\true{}}$, and similarly an outgoing edge and child BDD for $x=\false{}$. The two child BDDs can share some nodes and edges.
    \item Alternatively, a BDD for $\Phi$ may be a single root node $\{v\}$. In this case, we require that each function in $\Phi$ is a constant Boolean expression, and $L_V(v)$ is a function assigning each $\phi\in\Phi$ to its constant value in $\{\true{},\false\}$.
  \end{itemize}
 For convenience, when $\Phi=\{\phi\}$ is a singleton, we identify the label $\phi\mapsto b$ on each leaf with the constant value $b$ it assigns to the single formula $\phi$.\hfill(End of Def.~\ref{def:bdd}.)
\end{definition}

In the present work, we are interested in the number of variables that need to be probed in the worst-case, as reflected by the depth of the BDD.

\begin{definition}[Depth]
Given a BDD $G=(V,E,L_V)$, the \emph{depth} of the BDD is the maximal number of edges on a directed path in $G$ from the root node to a leaf. If $\card{V}=1$ then the depth is $0$. 

Given a set of Boolean expressions $\Phi\subseteq\name{Bool}[X]$, its depth is defined as the minimal depth of any BDD for $\Phi$.
\end{definition}

An immediate observation is that if we restrict ourselves to strategies that do not make obviously useless probes, i.e., do  not issue the same probe twice, then a BDD on $\Phi\subseteq\name{Bool}[X]$ has depth at most~$n=\card{X}$. We call a set of Boolean expressions as \emph{evasive} if its depth is~$n$, i.e., 
if there is no BDD with a better depth than that of the naive BDD querying all variables in order.

\begin{example}\label{ex:evasive} The AND function $x_1 \land \cdots \land x_n$ is evasive. Any decision tree for this function has height $n$, since in particular we must observe all the variables to establish that their conjunction is \true{}. OR is also evasive, by similar arguments.
\end{example}

\begin{example}\label{ex:nonevasive}
 $\psi_0= (w\land x) \lor  (x\land y) \lor (y\land z)$ is non-evasive. A strategy that first observes $x$ and $y$ will at most observe~3 variables: if $\funct{val}{x}=\funct{val}{y}=\true{}$, the evaluation is complete. Otherwise, assume w.l.o.g that $\funct{val}{x}=\false{}$; then~$w\land\false=\false{}$, so we do not need to observe~$w$. The same reasoning applies to~$z$ if $\funct{val}{y}=\false{}$.
\end{example}

\paragraph*{Problem definition}
We are given a set of Boolean expressions $\Phi\subseteq\name{Bool}[X]$ (where $\card{X}=n$). The \texttt{OPT-BDD-DEPTH} problem is defined as the optimization problem of finding a strategy that corresponds to a BDD of minimal depth for $\Phi$. The corresponding decision problem, \texttt{DEC-BDD-DEPTH}, is defined as deciding whether the depth of $\Phi$ is $\leq k$, for a given $k<n$. Finally, \texttt{DEC-BDD-EVASIVE} is the problem of deciding whether $\Phi$ is evasive.

Observe that we can reduce \texttt{DEC-BDD-EVASIVE}, in PTIME, to the complement of \texttt{DEC-BDD-DEPTH}, by deciding \texttt{DEC-BDD-DEPTH} for $k=n-1$, and taking the inverse. In contrast, \texttt{DEC-BDD-DEPTH} may not reduce in PTIME to
\texttt{OPT-BDD-DEPTH}, since this depends on the representation of the optimal strategy: the size of a BDD may be exponential in~$n$ so computing the value for \texttt{DEC-BDD-DEPTH} from the output of \texttt{OPT-BDD-DEPTH} is not in PTIME. 

\section{Results}\label{sec:results}
We next detail our technical results. We start by considering the general case of arbitrary Boolean expressions, and show that deciding \texttt{DEC-BDD-DEPTH} is generally intractable. We then turn to the case of read-once provenance, extensively studied in the context of, e.g., probabilistic databases~\cite{jha2013knowledge,roy2011faster,sen2010readonce}. We show that overall read-once Boolean provenance is always evasive. Then, we show a class of monotone Boolean expression that are ``far'' from evasive, and for which we can provide an optimal probing strategy. Finally, we identify a class of provenance expressions called monotone acyclic graph DNFs, for which we can decide evasiveness in PTIME.  
\subsection{General Provenance Expressions}\label{sec:negation}

We first observe that \texttt{DEC-BDD-DEPTH} is intractable when given a non-monotone Boolean formula, already for CNF and DNF.

\begin{proposition}\label{prop:depthhard}
Deciding \texttt{DEC-BDD-DEPTH}  is coNP-hard, even if the input Boolean expression is in DNF/CNF and the depth upper bound is $k=0$.
\end{proposition}

\begin{proof}
We first show this 
for CNFs. We reduce from the NP-hard problem of  deciding, given a CNF~$\phi$, if~$\phi$ is satisfiable, i.e., it is not the constant \false{} function. To do this, first check if~$\phi$ is falsifiable, which is easily done in PTIME by checking that it has one  non-trivial clause. If~$\phi$ is not falsifiable, answer that~$\phi$ is (vacuously) satisfiable. Otherwise, if~$\phi$ is falsifiable, then it is satisfiable iff it not constant, i.e., if its depth is greater than~0. We have thus reduced the NP-hard problem of Boolean satisfiability for a CNF to the complement of \texttt{DEC-BDD-DEPTH}, establishing that \texttt{DEC-BDD-DEPTH} is coNP-hard. Similarly, for DNF we can check in PTIME if the expression is satisfiable, and then the DNF is falsifiable iff the expression is not a constant~\true{}, i.e., its depth is~0.
\end{proof}

This also implies that the problem \texttt{OPT-BDD-DEPTH} is coNP-hard, since if we got an optimal 0-depth BDD (which, unlike general BDDs, can be compactly represented) we could decide that the depth is $\leq 0$. In contrast, this still does not imply the hardness of verifying evasiveness, which remains open.

\subsection{Read-once Provenance}\label{sec:ro}
Several works in the field of Probabilistic Databases~\cite{jha2013knowledge,roy2011faster,sen2010readonce} have studied queries that yield \emph{read-once} provenance (or provenance that can be compiled to read-once form), i.e., where each variable occurs at most once in each provenance expression. The reason is that read-once expressions enable easy computation of probabilities. The work of~\cite{drien2021managing} further lists cases when the provenance has \emph{overall} read-once form, i.e., when variables do not repeat across the provenance expressions of different output tuples.

In our context, we show that overall read-once sets of Boolean expressions are evasive. We call a read-once expression \emph{non-simplifiable} if it is either the constant \false{} or \true{}, or contains no occurrences of constants. Indeed, if that is not the case, then we can simplify the expression further by rules such as $\false{} \land \phi = \false{}$, yielding a smaller read-once expression.

\begin{example}\label{ex:ro}
Recall the provenance in Table~\ref{tab:res}. The second and third expressions are non-simplifiable read-once expressions, i.e., they contain no constants and no variable repetitions (also across these two expressions). The first expression is not read-once, as $a_0$ repeats in every term; but it is equivalent to the read-once expression $a_0\land((r_0\land e_0)\lor(r_1\land e_1)\lor (r_2 \land e_3))$; and similarly for the fourth expression. Finally, the entire query result is not overall read-once since variables such as $a_0, a_1, r_3$ repeat across Boolean expressions.
\end{example}

We can show the following by a recursion on the structure of the Boolean expression.

\begin{proposition}
Given a set of Boolean expressions $\Phi$ over a set of variables $X$, if $\Phi$ is overall read-once and each $\phi\in\Phi$ is non-simplifiable, then $\Phi$ is evasive.
\end{proposition}

Note an easy corollary: sets of Boolean expressions that are not overall read-once but are equivalent to an overall read-once set are also evasive, as evasiveness is preserved under equivalence. 

\begin{proof}
We prove this by induction on the number of variables in $\Phi$. The key observation is that for any non-simplifiable read-once Boolean formula~$\phi$ with~$n$ variables, for every variable $x$ of that formula, then one of $\phi_{x=\true{}}$ and $\phi_{x=\false{}}$ is a non-simplifiable read-once Boolean formula with $n-1$ variables. Indeed, if the one occurrence of~$x$ is positive, then if it is as part of a $\lor$-operation (resp., a $\land$-operation), then replacing $x$ by~\false{} (resp., by~\true{}) and simplifying yields a Boolean function that is still  read-once, is clearly non-simplifiable, and has $n-1$ variables. If the one occurrence of~$x$ is negative, then we do the same but replacing it respectively by~\true{} and~\false{} instead. Since $\Phi$ is overall read-once, we can perform this induction for each of its expressions separately.
\end{proof}

Read-once provenance is not necessarily evasive if 
variables may be repeated \emph{across} expressions: reconsider Example~\ref{ex:intro}, and observe that each expression is non-simplifiable read-once, and $x$ repeats across expressions. We have shown that probing one variable can always be avoided, i.e., the expression set is not evasive.

\subsection{Monotone Boolean Expressions}\label{sec:monotone}
Next, we consider monotone Boolean expressions, without negation. These expressions are of particular interest, since we can show a two-way correspondence between the provenance of SPJUs and the class of monotone Boolean $k$-DNF formulas, where every term (conjunction) consists of at most~$k$ (unique) variables.

\begin{proposition}\label{prop:UCQCharacterization}
\begin{enumerate}
    \item For each SPJU query~$Q$ there exists a value~$k$
    such that for every $X$-database $\bar{D}$, the provenance of each tuple in $Q(\bar{D})$ may be represented in a monotone $k$-DNF form. This form may further be constructed in PTIME in data complexity (i.e., as a function of~$\bar{D}$).
    \item Conversely, for every monotone $k$-DNF formula $\phi$, there exists an SPJ query $Q$ depending only on $k$, an $X$-database $\bar{D}$ whose size is linear in $\phi$ such that the query output $Q(\bar{D})$ is a singleton tuple whose provenance  is equivalent to $\phi$. 
\end{enumerate}
\end{proposition}

\begin{proof}[Proof (sketch)]
The first part of the proposition holds with~$k$ being the maximal number of relations joined by a conjunctive query within $Q$. To observe that this is the case, note that conjunctions in the provenance construction are associated with joins and disjunctions are associated with projection and union.

For the second part, we again exploit the correspondence between query and Boolean operations. Given a monotone $k$-DNF formula~$\phi$, consider a DB $\bar{D}$ with two relations. Relation $R$ encodes the variables of~$\phi$, where for each variable~$x$ in~$\phi$ we have a tuple~$R(x)$ (using $x$ as a value) annotated by~$x$. Relation~$S$ encodes the terms of~$\phi$ such that for each term $x_1\wedge x_2\wedge\dots\wedge x_k$ in~$\phi$ we have a tuple $S(x_1,x_2,\dots,x_k)$ (using $x_1, x_2,\dots, x_k$ as values) and annotated by $x_1$. If a term is of size $<k$ we can repeat one of its variables to obtain an equivalent term of size exactly~$k$. The query $Q$ is a binary CQ fixed for $k$ (having no unions, using only equality joins and a projection on all variables): $\name{ans}():-S(z_1,\dots,z_k), R(z_1), \dots, R(z_k)$. By this construction, each tuple in the join result corresponds to a term in~$\phi$ and has a provenance of the form of a conjunction $x_1\wedge x_2\wedge\dots\wedge x_k$, where every $x_i$ stands for an original variable, and the repeated occurrence of $x_1$ by the join with $S$ is absorbed. Then projecting out all the variables yields the disjunction of these conjunctions, which is equivalent to~$\phi$. 
\end{proof}

The hardness result of Prop.~\ref{prop:depthhard} does not immediately apply to the case of monotone functions, which we leave open. 
Let us instead study bounds that we can obtain in special cases. 
For instance, to prove that an expression in DNF is \true{}, we need to prove, for one of its terms (conjunctions), that each of its variables are \true{}; and similarly for CNF clauses for proving \false{}. Therefore, the depth of a monotone expression is bounded from below by the maximum size (number of distinct variables) of a term in its DNF and of a clause in its CNF (assuming the expressions are simplified to avoid subsumed terms/clauses). As stated by the following theorem, there exists a class of Boolean expressions that have a depth linear in this bound, and exponentially smaller than the number of variables.

\begin{theorem}\label{thm:logoptimal}
For arbitrarily large integers~$n$, there is a monotone Boolean expression of size~$\Theta(n)$, such that in DNF its term size (number of distinct variables in a conjunction) is at most~$k=\oof{\log{n}}$, and its depth is $2k-1=\oof{\log{n}}$.

\end{theorem}
\begin{proof}[Proof (sketch)]
  Define the monotone DNF formula $\psi_0 = (w\wedge x)\vee(x\wedge y)\vee(y\wedge z)$ (as in Example~\ref{ex:nonevasive}) and recursively define $\psi_{i+1}= (u_i\wedge \psi_i) \vee (u_i\wedge v_i) \vee (v_i\wedge \psi'_{i})$ where $u_i$, $v_i$ are fresh variables and $\psi'_{i}$ is obtained by consistently replacing the variables of $\psi_i$ by fresh variables. In DNF, the resulting term size is $k=i+2$. We can show that $\card{\vars{\psi_i}}$ is exponential in $i$ since we double the expression size at each level, and that there exists a BDD $\psi_i$ whose expected cost is $\oofi{i}$ (by probing first the $u_i, v_i$ variables).
\end{proof}

The theorem implies that exists a class of Boolean expressions that are ``far'' from evasive, namely, the optimal strategy probes only a logarithmic number of variables in the worst case. Moreover, the optimal number of probes for this class is close to the overall lower bound on the number of probes  for any $k$-DNF.
Our proof is constructive: not only do we show the existence of an optimal strategy, but the proof also provides the strategy itself.

\subsection{Monotone Acyclic Graph DNFs}\label{sec:graph}
As a first step towards understanding monotone DNFs in general, we next focus on a subclass of Boolean  expressions which we call \emph{monotone acyclic graph DNFs}. The subclass is defined as follows:

\begin{definition}[Monotone acyclic graph DNFs]
A monotone graph DNF $\phi$ is a monotone DNF where each term includes at most two variables (i.e., 2-DNF). We identify $\phi$ to a graph $G$ whose nodes are the variables of $X$ and which has one edge $\{x,y\}$ per term $x\land y$. 
We say that $\phi$ is \emph{acyclic} if $G$ is. For convenience, we will often abuse notation and identify $\phi$ and $G$ (but paying attention to the fact that $G$ does not represent terms with one single variable), and we will often consider $G$ as a rooted tree, by picking some root.
\end{definition}
 
\begin{example}\label{ex:graphdnf} The expression $\psi_0= (w\land x) \lor  (x\land y) \lor (y\land z)$, shown in Ex.~\ref{ex:nonevasive} to be non-evasive, is a monotone acyclic graph DNF. 
\end{example}

Characterizing which queries guarantee this form of provenance is left open. However, given provenance of this shape, we can show the following.

\begin{theorem}
  \label{thm:acyclic}
  Given a monotone acyclic graph DNF, deciding \texttt{DEC-BDD-EVASIVE} is in PTIME.
\end{theorem}

We prove this theorem in the rest of the section.
The high-level idea is that we identify a \emph{non-evasiveness pattern}, defined recursively on the structure of the graph DNF. We explain that finding such a pattern is in PTIME, and that such a pattern is present iff the formula is non-evasive,
leading to the conclusion of Theorem~\ref{thm:acyclic}.

Our theorem holds for any acyclic monotone graph DNF, but for convenience, we will eliminate some anomalies that can be handled separately (and in PTIME), by assuming:
\begin{itemize}
    \item All the variables of $X$ occur in~$\phi$ (otherwise, it is trivially non-evasive).
    \item The graph is connected (otherwise, the formula is evasive iff each connected component is).
    \item There are no subsumed terms, i.e., if we have a singleton term $x$, then there are no terms $x \land y$ (as these can be removed).
\end{itemize}

We now present the pattern that characterizes non-evasiveness:

\begin{definition}
  Let $\phi$ be a monotone acyclic graph DNF, and let $x$ be a variable. Let $T$ be the graph of~$\phi$ represented as a tree rooted at~$x$. For $y$ a variable of~$\phi$, we denote by $\phi_{x:y}$ the monotone acyclic graph DNF obtained by restricting $\phi$ to the variables in the subtree of~$T$ rooted at~$y$, and the terms involving only these variables.
  
  A \emph{non-evasiveness pattern} $\Pi$ 
  rooted at $x$ for $\phi$ is a labeled tree inductively defined as follows:
  \begin{itemize}
    \item If $x$ is a variable that does not appear in any term, then $\Pi$ is a single
      leaf node $n$ labeled by~$x$.
    \item Otherwise, $\Pi$ has a root node $n$ labeled with~$x$ and with the
      following children: for each variable $y$ co-occurring in a term with
      $x$ (i.e., every child of~$x$ in~$T$), 
      we choose a grandchild $w_y$ of~$y$ (we require that such a grandchild exists), and $n$ has a child node consisting of a non-evasiveness pattern rooted at $w_y$ for $\phi_{x:w_y}$.
  \end{itemize}
\end{definition}

\begin{example}

The formula $\psi_0$ in Example~\ref{ex:nonevasive} has a non-evasiveness pattern, for instance the one with root labeled $x$ and with one child labeled $z$. Generalizing from this example to the ``path-shaped'' monotone acyclic graph DNFs of the form $\psi_n = (x_0 \land x_1) \lor (x_1 \land x_2) \lor \cdots \lor (x_{n-1} \land x_{n})$, there is a non-evasiveness pattern iff $n$ is divisible by~$3$. $\psi_0$ is the path expression with $n=3$, and is indeed non-evasive.
\end{example}

The existence of a non-evasiveness pattern for~$\phi$ can be checked
in PTIME recursively: we successively root $\phi$ at each variable $x$ and compute bottom-up,
for each variable $y$, if there is a non-evasiveness pattern rooted at~$y$. Note also that the choice of root of a non-evasiveness pattern $\Pi$ among the variables occurring in~$\Pi$ is inessential: for any node $y$ labeling a node of~$\Pi$, we can re-root $T$ and $\Pi$ at $y$ and obtain a non-evasiveness pattern rooted at~$y$.

We will show that a non-evasiveness pattern for $\phi$ exists iff $\phi$ is non-evasive. We first show the easy direction:

\begin{lemma}\label{lem:patterntononevasive}
  Given a monotone acyclic graph DNF $\phi$, if $\phi$ has a non-evasiveness pattern $\Pi$, then it
  is not evasive.
\end{lemma}

\begin{proof}
  We show the claim by induction on the number of nodes of~$\Pi$. If $\Pi$
  is just a singleton node, then $\phi$ is a formula where $x$ does not appear in any term, so does not need to be queried.  Otherwise, let $x$ be the variable labeling the root node~$n$ of~$\Pi$,
  let $y_1, \ldots, y_m$ be the children of~$x$ in~$T$ (with $m \geq 1$), and let $w_1, \ldots, w_m$ be the variables labeling the corresponding child nodes of~$n$ in~$\Pi$; let $z_1, \ldots, z_m$ be their respective parents (each $z_i$ is a child of $y_i$). 

  Let us build a BDD witnessing that $\phi$ is not evasive by asking
  about all the variables $y_1, \ldots, y_m$. If they all evaluate to~$\false$, then all terms involving~$x$ are falsified, so we need not query~$x$ and the remaining depth is less than the number of remaining variables.
  Hence, to conclude, it suffices to argue that the formula is not evasive when some $y_i$ evaluates to true, say $\phi_{y_{i_0} = \true}$. 

  Now, in $\phi_{y_{i_0} = \true}$, the term $y_i \land z_i$ simplifies to the
  singleton term $z_i$, which subsumes the term $z_i \land w_i$.
  Thus, after performing partial evaluation, the resulting formula has a connected component that is precisely $\phi_{z_i:w_i}$.   Now, the subtree of~$\Pi$ rooted at~$n_i$ is by definition 
  a non-evasiveness pattern for $\phi_{z_i:w_i}$ rooted at~$w_i$.
  Hence, by induction hypothesis, we know that $\phi_{z_i:w_i}$  is not evasive, 
  so $\phi_{y_{i_0} = \true}$ has a connected component that is not evasive and is itself non-evasive, concluding the proof. 
\end{proof}

We also claim that the converse holds, via:

\begin{lemma}\label{lem:nopatterntoevasive}
  If a monotone acyclic graph DNF $\phi$ has no non-evasiveness pattern, then for any variable $x$ of~$\phi$, there exists
  $b \in \{\true{}, \false{}\}$ such that, after performing partial evaluation on~$\phi_{x=b}$, no connected component has a non-evasiveness pattern.
\end{lemma}

This implies that if $\phi$ does not have a non-evasiveness pattern, then it is
evasive. Indeed, applying repeatedly the above lemma justifies that any BDD has a branch leading always to a formula without a non-evasiveness pattern,
until we reach a formula with no variables (which is constant and evasive).
Let us prove the lemma:

\begin{proof}
Root $\phi$ at the variable $x$, yielding a tree~$T$. We say that
  a variable $y$ of~$T$
  is \emph{special} if there is a non-evasiveness pattern for $\phi_{x:y}$ rooted at~$y$.
  Because we assumed $\phi$ has no subsumed terms,
  all leaves of $T$ are special, and we know by definition of a non-evasiveness pattern that for every special
  variable $x$ and child~$y$ of~$x$ there is a grandchild of~$y$ that is also
  special. We make the following observation:

  \begin{observation}
    \label{obs:badcase}
    If $x$ has a special child and a special grandchild, 
    then the whole graph DNF~$\phi$ has a non-evasiveness pattern in which $x$ does not appear.
  \end{observation}

  \begin{proof}
    Let $y$ be the special child and $z$ the special grandchild of~$x$, and let $\Pi_y$ and $\Pi_z$ be respectively the non-evasiveness pattern of $\phi_{x:y}$ rooted at $y$ and the non-evasiveness pattern of~$\phi_{x:z}$ rooted at~$z$.
    We proceed by induction on the height of~$T$. The base case is when~$T$ has height at most~$3$: then the special nodes $z$ and $y$ cannot have any child and must be leaves, and we have a suitable non-evasiveness pattern for $\phi$ with one variable, say $y$, as the root, and the other as its only child.

    For the induction case, there are two subcases. First, if
    $z$ is not a child of~$y$ in~$T$, then we conclude that $\phi$ has a
    non-evasiveness pattern rooted at~$y$. We build it from the non-evasiveness pattern $\Pi_y$ rooted at~$y$ for $\phi_{x:y}$, but we take into account the additional child $x$ of~$y$ when rooting~$\phi$ at~$y$ by adding $\Pi_z$ as a child node of the root in~$\Pi_y$. Note that $x$ does not occur as a node label in the result.
    
    Second, if $z$ is a child of~$y$, then~$y$ is not a leaf. As $y$ is special and $z$ is a neighbor of~$y$, the definition of~$\Pi_y$ ensures that 
    there must be a grandchild~$y'$ of~$z$ that is special. Letting $w$ be the
    intermediate node between $z$ and $y'$, it witnesses that~$z$ is not a leaf, so by definition of~$\Pi_z$ there must be a
    grandchild~$z'$ of~$w$ that is also special. We can thus repeat the argument
    on the subtree of~$T$ rooted at~$w$, which is of strictly smaller height. By
    induction, this subtree has a non-evasiveness pattern $\Pi_w$ that does not
    involve~$w$. Now, $\Pi_w$ is a non-evasiveness pattern for the whole~$\phi$,
    because we can connect the rest of the tree~$T$ to~$w$ without breaking the
    existing non-evasiveness pattern (as it does not involve~$w$).
  \end{proof}

  We now prove the lemma by contradiction. Assume there is a variable $x$ such that for each $b \in \{\false, \true\}$ some connected component of~$\phi_{x=b}$ has a non-evasiveness pattern, and let us show this implies $\phi$ has a non-evasiveness pattern. Let~$T$ be the tree obtained by rooting~$\phi$ at~$x$.  We wish to show that $x$ has a special child and a special grandchild, to conclude using the previous observation.
  
  Let us first take $b=\false$. Let us consider the monotone acyclic graph DNF $\phi_{x=\false}$, perform partial evaluation to remove $x$ and all terms involving $x$, and
  root the child subtrees $T_1, \ldots, T_n$ of $T$ at the variables $y_1, \ldots, y_n$ that co-occur with~$x$ in~$\phi$: these are the connected components to consider. We know by hypothesis that some connected component $T_i$ of~$\phi_{x=\false}$ has a
  non-evasiveness pattern $\Pi_i$. If $\Pi_i$ does not use $y_i$ as a label, then we can use it as non-evasiveness pattern for all of $\phi$, because adding back the rest of the tree $T$ does not break the definition of a non-evasiveness pattern (it does not add any new neighbor to consider), and this immediately concludes. Hence, it suffices to consider the case where $\Pi_i$ uses $y_i$ as a label, and we can re-root it to a non-evasiveness pattern of~$T_i$ rooted at~$y_i$. Hence, $y_i$ is special in~$T$.

  Now consider the monotone acyclic graph DNF $\phi_{x=\true}$. Performing partial evaluation, the neighbors $y_1, \ldots, y_n$ become singleton terms, and the other terms involving them are subsumed. The singleton connected components corresponding to the $y_1, \ldots, y_n$ cannot have non-evasiveness patterns because they appear in singleton terms. The other connected components are the subtrees $T_1, \ldots, T_m$ rooted at the grandchildren $z_1, \ldots, z_m$ of the root of~$T$. By hypothesis, one of them, say $T_j$,
  has a non-evasiveness pattern, and as in the previous
  paragraph has a non-evasiveness pattern rooted at the grandchild~$z_j$, so that
  $z_j$ is special in~$T$.  We can now use the existence of the special child~$y_i$ of~$T$ and the special
  grandchild~$z_j$ of~$T$ to establish using Observation~\ref{obs:badcase}
  that~$T$ has a non-evasiveness pattern. This concludes the proof.
\end{proof}

To conclude, we have shown in Lemma~\ref{lem:patterntononevasive} that indeed, if a non-evasiveness pattern occurs in a monotone acyclic graph DNF $\phi$, we have a strategy that always probes less variables than the number of variables of~$\phi$, i.e., $\phi$ is non-evasive. Lemma~\ref{lem:nopatterntoevasive} proves, for the other direction, that if $\phi$ has no non-evasiveness pattern, then it is indeed evasive.
This completes the proof of Theorem~\ref{thm:acyclic}.
\section{Related Work}\label{sec:related}

\paragraph*{Interactive Boolean Evaluation and Consent Management} The problem of Interactive Boolean Evaluation~\cite{allen2017evaluation,boros2000sequential,deshpande2014approximation} and BDD optimization has been extensively studied in multiple contexts. These include system testing, e.g.,~\cite{boros2000sequential,unluyurt2004sequential} 
(where it is called \emph{Sequential System Testing} or \emph{Sequential Diagnosis}), BDD design, e.g.,~\cite{cicalese2014diagnosis,fiat2004decision,keren2008reduction} (where it is also called \emph{Discrete Function Evaluation}), active learning~\cite{golovin2011adaptive} (where it is a particular case of \emph{Bayesian Active Learning}) and its connection to other problems such as Stochastic Set Cover~\cite{deshpande2014approximation,kaplan2005learning} 
(where it is termed \emph{Stochastic Boolean Function Evaluation (SBFE)}). Boolean Evaluation is also a component of \emph{Consent Management}~\cite{amsterdamer2019pepper,drien2021managing}, where the choice of variables to observe coincides with requests to attain consent to use specific input tuples.
Many of these studies focus on optimization goals different than ours, most notably the expected number of observed variables under a probabilistic model~\cite{allen2017evaluation,amsterdamer2019pepper,boros2000sequential,cicalese2014diagnosis,deshpande2014approximation,drien2021managing,kaplan2005learning,keren2008reduction,unluyurt2004sequential}, corresponding to the \emph{expected} path length in a BDD; but other alternatives such as minimizing the size of the BDD have been considered~\cite{fiat2004}. 

Worst-case/depth analysis of BDDs and related problems have been studied in additional contexts, e.g., testing edges to decide about graph properties (e.g.,~\cite{chronaki1990survey,kahn1983topological,scheidweiler2013lower,yao1988monotone}) or letters to decide string properties~\cite{chronaki1990survey}. The Boolean functions considered are thus very different from ours. There is also work on computing BDDs with minimum depth for input-output sample pairs (e.g.,~\cite{cicalese2014diagnosis,fiat2004decision}), but where a Boolean formula is not a-priori known.

\paragraph*{Probabilistic databases} A large body of work studied query evaluation in probabilistic databases~\cite{dalvi2007efficient,koller1999probabilistic,papaioannou2018supporting,suciu2011probabilistic,vandenbroeck2017query}, which resembles our work in the use of data provenance~\cite{green2007provenance,imielinski1984incomplete} to track the connection between input and query output and the approach of identifying classes of queries that yield provenance expressions of favorable structure. 
In particular, prior work on probabilistic databases has characterized query classes where provenance may be transformed into read-once form for its useful properties such as easy probability computation~\cite{fink2011optimal,jha2013knowledge,roy2011faster,sen2010readonce}. In our context, we have studied the evasiveness of such expressions (Section~\ref{sec:ro}).

\paragraph*{Data Provenance} Provenance for query results has been extensively studied, with multiple models and applications~\cite{amsterdamer2011limitations,amsterdamer2011aggregate,amsterdamer2011putting,buneman2001why,cheney2009provenance,deutch2014circuits,glavic2009provenance,glavic2009perm,green2007provenance,imielinski1984incomplete,karvounarakis2010querying,olteanu2012factorised,senellar2018provsql}. 
Specifically, as discussed above the shape of the Boolean provenance may depend on the used data model (e.g., relational, graph databases, nested data) and query language (positive relational algebra, difference, Datalog). This, in turn, may affect the problem analysis, as we have shown for monotone Boolean provenance yielded by queries without negation.

\paragraph*{Crowdsourced databases} In practical scenarios, the worst-case number of observed variables may be translated to the worst-case number of questions posed to a human expert or crowd members, which should be accounted for, e.g., in optimization or by proper budget assignment. Works on crowd data sourcing such as~\cite{bergman2015query,franklin2011crowddb,li2018cdb,marcus2011crowdsourced,parameswaran2012deco,parameswaran2011human} considered different problems that typically also involve some optimization of the questions posed to crowd members. Specifically, the work of~\cite{bergman2015query} considered the cleaning of an input database in order to fix query results, which involves the propagation of Boolean correctness from the input to output tuples; but the problem they solve is completely different and hence so are their analysis and results.

\section{Conclusion and Future Work}\label{sec:conc}
We have studied in this paper the problem of \emph{evaluating} provenance expressions, namely repeatedly probing individual variables to reveal their truth values in order to reveal the truth value of entire expressions. We have focused here on a worst-case analysis of the problem based on BDD representation of the provenance. We have shown that computing the number of required probes is generally intractable, but identified several classes of provenance expressions for which the problem becomes tractable. In some of these classes, we can identify that the expressions are evasive and one cannot do better than simply probing all variables, and in others, we can in fact do exponentially better than this naive approach. 

The motivation for our study originates in our work on consent management~\cite{drien2021managing} where we have used interactive Boolean evaluation of provenance expressions. Yet our techniques and analysis may have applications in other domains as well, such as data cleaning and view maintenance.
In future work we will further explore such application domains, and will continue characterizing the query classes and corresponding classes of Boolean expressions that admit efficient interactive evaluation algorithms. 

\paragraph*{Acknowledgements}
We thank the anonymous reviewers for their insightful comments. This work was partly funded by the Israel Science Foundation (grant No. 2015/21), by a grant from the Israel Ministry of Science and Technology, and by the grants ANR-18-CE23-0003-02 (``CQFD'') and ANR-19-CE48-0019 (``EQUUS'').

\bibliographystyle{abbrv}
\bibliography{bib}

\end{document}